\documentclass{article}

\usepackage[english]{babel}
\usepackage{graphicx}
\usepackage{textcomp}
\usepackage{amssymb}
\usepackage{amsmath}
\usepackage{setspace}
\usepackage{geometry}
\usepackage{bm}
\usepackage{centernot}
\usepackage{hyperref}

\newtheorem{proposition}{Proposition}
\newtheorem{corollary}{Corollary}
\newtheorem{remark}{Remark}
\newtheorem{definition}{Definition}

\newtheorem{lemma}{Lemma}

\newenvironment{proof}[1][Proof:]{\begin{trivlist}
\item[\hskip \labelsep {\bfseries #1}]}{\end{trivlist}}

\title{Global hyperbolicity and factorization in cosmological models}
\author{Z. Avetisyan\footnote{z.avetisyan@math.ucsb.edu}}
\date{Department of Mathematics, UCSB, Santa Barbara, USA}

\begin{document}

\maketitle

\abstract{The subject of the paper is the geometry and topology of cosmological spacetimes and vector bundles thereon, which are used to model physical fields propagating in the universe. Global hyperbolicity and factorization properties of the spacetime and the vector bundle that are usually independently assumed to hold, are now derived from a minimal set of assumptions based on the recent progress in differential geometry and topology.}

\section*{Introduction}

We deal with global hyperbolicty and factorization of spacetimes and vector bundles that appear in mathematical models of homogeneous cosmology. That a globally hyperbolic spacetime $M$ is topologically factorized as $M\simeq\mathbb{R}\times\Sigma$ has been known since decades (see Chapter 3 in \cite{BeemEhrlichEasley}), but that it can be factorized geometrically as $(M,\mathrm{g})\simeq(\mathbb{R}\times\Sigma,\beta^2\oplus-h_*)$, with $t\mapsto h_t$ being a smooth family of Riemannian metrics, is a relatively new result \cite{BernalSanchez2005}. The obvious question of the converse implication, i.e., which spacetimes $(\mathbb{R}\times\Sigma,\beta^2\oplus-h_*)$ are globally hyperbolic, does not seem to have a complete answer for $\Sigma$ non-compact (cosmologically most relevant). For homogeneous and isotropic (FRW) cosmological spacetimes the answer is known due to the fact that these spacetimes are warped products \cite{BeemEhrlichEasley}. But purely homogeneous cosmological spacetimes are not warped products, and their global hyperbolicity does not seem to have been fully settled in the literature. In the first section of this paper we use the recently developed idea of continuous Riemannian metrics on smooth manifolds to establish convenient necessary and sufficient conditions for a constant time hypersurface $\{t\}\times\Sigma$ to be a Cauchy hypersurface in the spacetime $(\mathbb{R}\times\Sigma,\beta^2\oplus-h_*)$. This immediately provides the global hyperbolicity of homogeneous cosmological spacetimes, as shown in the second section. Note that for homogeneous cosmological models one does not need to indulge into non-smooth metrics, but the latter is useful for other spacetimes with metrically incomplete constant time hypersurfaces, where the standard completeness arguments break down, at least in their old interpretation.

Another factorization $(M,\mathrm{g})\simeq(\mathbb{R}\times\Sigma,\beta^2\oplus-h_*)$ arises from the assumption of spatial homogeneity in cosmology, i.e., that a Lie group acts by isometries transitively on spatial hypersurfaces \cite{StephaniKramerMacCallumHoenselaersHerlt2003}. The natural question is whether these two factorizations coincide, yielding a convenient foliation by homogeneous Cauchy hypersurfaces. To our knowledge, the literature in mathematical cosmology mostly assumes this to be the case in their models, usually in combination with other convenient properties, and it is often difficult to disentangle the logical interdependencies of various assumptions used together. In the second section of this work we set forth a minimal set of assumptions on a homogeneous cosmological spacetime, and show how the rest of the desired properties follow automatically.

Finally, once the spacetime is factorized as $\mathbb{R}\times\Sigma$, one would like to know whether a vector bundle $\mathcal{T}\to\mathbb{R}\times\Sigma$ is also factorized as $\mathcal{T}\simeq\mathbb{R}\times\mathcal{S}$ where $\mathcal{S}\to\Sigma$ is a fixed vector bundle. This is automatic for covariant vector bundles (tensor, spinor) but appears to be unanswered for an abstract vector bundle. In the third section we apply recent results from geometric topology to give an affirmative answer to this question in full generality. This factorization becomes indispensable for separation of variables techniques on PDEs acting on sections of vector bundles.

\section*{Global hyperbolicity}

A standard assumption that ensures the dynamical consistency of a field theory coupled with classical gravity is that the spacetime is globally hyperbolic. A smooth connected $d+1$-dimensional Lorentzian manifold $(M,\mathrm{g})$ is globally hyperbolic if it admits a Cauchy hypersurface. Superficially, this seems to be a very mild condition about the special nature of time in a spacetime, that does not explicitly violate general covariance and should not result in a remarkable loss of generality. However, it turns out that global hyperbolicity implies rather strict geometrical and topological restrictions, essentially reducing the most general case to a classical archetype. In fact, by the famous Theorem 1.1 in \cite{BernalSanchez2005} there exists a smooth temporal function $t\in C^\infty(M,\mathbb{R})$ and a smooth connected $d$-dimensional manifold $\Sigma$ such that $M$ is isometrically diffeomorphic to $\mathbb{R}\times\Sigma$ taken with the Lorentzian metric $\beta^2\oplus-h_*$, where $\beta\in C^\infty(\mathbb{R}\times\Sigma,\mathbb{R}_+)$ is a positive smooth function and $\mathbb{R}\ni t\mapsto h_t\in C^\infty(\mathrm{T}^{*\otimes2}\Sigma)$ is a smooth family of Riemannian metrics on $\Sigma$. Therefore there is no loss of generality in assuming that the Lorentzian manifold $(M,\mathrm{g})$ is $(\mathbb{R}\times\Sigma,\beta^2\oplus-h_*)$. In this section all manifolds are assumed to be without boundary.

Several partial converses of this statement can be found in the literature, perhaps the best known being Theorem 3.68 (with variations) in \cite{BeemEhrlichEasley}, proving the global hyperbolicity of a warped product with a complete Riemannian manifold. A spacetime of the form $(\mathbb{R}\times\Sigma,\beta^2\oplus-h_*)$ need not be a warped product, however, and completeness of Riemannian hypersurfaces is not necessary for global hyperbolicity. Moreover, purely homogeneous cosmological spacetimes are not warped products, and there seems to be no widely known published proof that these spacetimes are globally hyperbolic. Among the results that we know of, the closest to an answer comes \cite{ChoquetBruhatCotsakis2002} where authors consider the particular case which, in the terminology of the present paper, corresponds to a globally bounded function $D$ to be defined below. In this section we present a necessary and sufficient criterion of global hyperbolicity that will allow us to easily establish it for cosmological spacetimes. It also covers situations where the Cauchy surfaces are not necessarily complete, such as a bounded open causal diamond in another globally hyperbolic spacetime, which is always globally hyperbolic.

Let the spacetime $(\mathbb{R}\times\Sigma,\beta^2\oplus-h_*)$ be given, and choose a continuous (not necessarily smooth) Riemannian metric $h_\infty$ on $\Sigma$ such that the Riemannian manifold $(\Sigma,h_\infty)$ is complete as a metric space. There is always a good supply of smooth such metrics on every connected (second countable) manifold by Theorem 1 in \cite{NomizuOzeki1961}, but we will allow for merely continuous metrics. A beautiful recent treatment of continuous Riemannian metrics can be found in \cite{Burtscher2015}. The following is a minimal adaptation of Lemma 3.65 in \cite{BeemEhrlichEasley} to continuous Riemannian metrics.

\begin{lemma}\label{FinLLemma} Let $h_\infty$ be a continuous Riemannian metric on the connected manifold $\Sigma$ such that $(\Sigma,h_\infty)$ is complete. Let further $\gamma:[0,1)\to\Sigma$ be an absolutely continuous curve of finite length. Then
$$
\exists\lim\limits_{t\to1-}\gamma(t).
$$
\end{lemma}
\begin{proof} Let $\mathrm{L}$ and $\mathrm{d}$ be the length structure and distance induced by $h_\infty$ and consider the closed ball
$$
K\doteq\overline{\mathcal{B}(\gamma(0),\mathrm{L}(\gamma))}=\left\{x\in\Sigma\,{\big|}\quad d(\gamma(0),x)\le\mathrm{L}(\gamma)\right\}.
$$
Since $\mathrm{d}(\gamma(0),\gamma(t))\le\mathrm{L}(\gamma|_{[0,t]})<\mathrm{L}(\gamma)$ for all $t\in[0,1)$, we see that $\gamma([0,1))\subset K$. By Proposition 4.1 in \cite{Burtscher2015}, the topology induced by $\mathrm{d}$ coincides with the manifold topology on $\Sigma$, which is locally compact. Then by Proposition 2.5.22 in \cite{BuragoBuragoIvanov2001} we find that $K$ is compact, and thus there exists a sequence $\{t_n\}_{n=1}^\infty\subset[0,1)$ with $t_n\to 1-$ such that $\gamma(t_n)\to x_0\in\Sigma$. Assume towards a contradiction that $\gamma(t)\not\to x_0$ as $t\to1-$, so that a sequence $\{s_n\}_{n=1}^\infty\subset[0,1)$ exists together with an $\epsilon_0>0$ such that
$\mathrm{d}(\gamma(s_n),x_0)>\epsilon_0$ and $s_n>t_n$ for all $n\in\mathbb{N}$. By passing to a subsequence, if necessary, we can assume that $\mathrm{d}(\gamma(t_n),x_0)<\epsilon_0/2$ and $t_{n+1}>s_n$ for all $n\in\mathbb{N}$. It follows that
$$
\mathrm{L}(\gamma)\ge\sum\limits_{n=1}^\infty\mathrm{d}(\gamma(t_n),\gamma(s_n))\ge\sum\limits_{n=1}^\infty\frac{\epsilon_0}2=\infty,
$$
contrary to the finiteness of $\mathrm{L}(\gamma)$. This contradiction implies the assertion of the lemma. $\Box$
\end{proof}

Define the function $D:\mathbb{R}\times\Sigma\to\mathbb{R}_+$ by
\begin{equation}
D(t,x)\doteq\frac{\beta^2(t,x)}{\min\limits_{\substack{X\in\mathrm{T}_x\Sigma\\h_\infty[x](X,X)=1}}h_t[x](X,X)}>0,\quad\forall (t,x)\in\mathbb{R}\times\Sigma.\label{DDef}
\end{equation}
The quantity in the denominator of the above formula is essentially the lowest eigenvalue of the matrix $h_t[x]$ in the orthonormal frame relative to $h_\infty[x]$. Now $D$ is clearly a continuous function. Recall that $\mathrm{J}^\pm(t,x)\subset\mathbb{R}\times\Sigma$ denotes the causal future/past of the point $(t,x)\in\mathbb{R}\times\Sigma$. Denote for convenience
$$
\Sigma_t\doteq\{t\}\times\Sigma\subset\mathbb{R}\times\Sigma,\quad\forall t\in\mathbb{R}.
$$
\begin{proposition}\label{0xSigmaCauchyProp} The hypersurface $\Sigma_0$ is a Cauchy surface if and only if
\begin{equation}
\sup\limits_{\mathrm{J}^+(t,x)\cap(\mathbb{R}_-\times\Sigma)}D<\infty,\quad\sup\limits_{\mathrm{J}^-(t,x)\cap(\mathbb{R}_+\times\Sigma)}D<\infty,\quad\forall (t,x)\in\mathbb{R}\times\Sigma.\label{0xSigmaCauchyProp:1}
\end{equation}
\end{proposition}
\begin{proof} If $\Sigma_0$ is a Cauchy hypersurface then both regions $\mathrm{J}^\pm(t,x)\cap(\mathbb{R}_\mp\times\Sigma)$ are compact (say, Proposition 6.6.6 in \cite{HawkingEllis1973}) for every point $(t,x)$, so that the continuous function $D$ is bounded on them. Assume now that $D$ is bounded on all such regions. We will prove that every $C^1$ inextendible causal curve intersects $\Sigma_0$ exactly once. Let $\gamma:(a,b)\to\Sigma$ be a $C^1$ curve such that $(a,b)\ni t\mapsto(t,\gamma(t))\in\mathbb{R}\times\Sigma$ is causal,
\begin{equation}
[\beta^2\oplus-h_*]\bigl((1,\dot\gamma(t)),(1,\dot\gamma(t))\bigr)=\beta^2(t,\gamma(t))-h_t[\gamma(t)](\dot\gamma(t),\dot\gamma(t))\ge0,\quad\forall t\in(a,b),\label{gammaCausal}
\end{equation}
and let the possibly infinite interval $(a,b)\subseteq\mathbb{R}$ be the maximal domain of $\gamma$. All we need to show is that $0\in(a,b)$. Without loss of generality, assume towards a contradiction that $a<t_0<b\le0$. This means that
\begin{equation}
\nexists\lim\limits_{t\to b-}\gamma(t).\label{gammaLimNoexist}
\end{equation}
From
$$
\left\{(t,\gamma(t))\,{\big|}\quad t\in[t_0,b)\right\}\subset\mathrm{J}^+(t_0,\gamma(t_0))\cap(\mathbb{R}_-\times\Sigma)
$$
and (\ref{0xSigmaCauchyProp:1}) we find that
$$
D(t,\gamma(t))\le D_0<\infty,\quad\forall t\in[t_0,b).
$$
Using this together with (\ref{DDef}) and (\ref{gammaCausal}) we see that
$$
h_\infty[\gamma(t)](\dot\gamma(t),\dot\gamma(t))\le D(t,\gamma(t))\frac{h_t[\gamma(t)](\dot\gamma(t),\dot\gamma(t))}{\beta^2(t,\gamma(t))}\le D_0,\quad\forall t\in[t_0,b).
$$
It follows that
$$
\int\limits_{t_0}^b\sqrt{h_\infty[\gamma(t)](\dot\gamma(t),\dot\gamma(t))}dt\le\sqrt{D_0}(b-t_0).
$$
This means that the restricted curve $\gamma:[t_0,b)\to\Sigma$ has finite length in the complete Riemannian manifold $(\Sigma,h_\infty)$, and therefore by Lemma \ref{FinLLemma} the limit of $\gamma(t)$ as $t\to b-$ exists, in contradiction to (\ref{gammaLimNoexist}). This implies that $b>0$, and a similar argument will show that $a<0$. The proof is complete. $\Box$
\end{proof}

\begin{remark} The choice of $\Sigma_0$ in the above proposition is arbitrary, and the same arguments apply to every $\Sigma_t$.
\end{remark}

\begin{corollary} A spacetime $(M,\mathrm{g})$ is globally hyperbolic if and only if it is isometrically diffeomorphic to $(\mathbb{R}\times\Sigma,\beta^2\oplus-h_*)$ where $\beta$ and $h_*$ satisfy (\ref{0xSigmaCauchyProp:1}).
\end{corollary}

The next logical question is the choice of a practically convenient complete metric $h_\infty$ on $\Sigma$ so as to produce a function $D$ by (\ref{DDef}) that satisfies (\ref{0xSigmaCauchyProp:1}). If the global behaviour of $h_t$ as a function of $t$ is uniform in $x$ in a certain sense then the following sufficient condition can be found.

\begin{proposition}\label{COmegaProp} Suppose that a continuous positive function $C\in C(\mathbb{R},\mathbb{R}_+)$ exists such that
$$
\Omega^2(x)\doteq\inf\limits_{t\in\mathbb{R}}\left[\frac{C(t)}{\beta^2(t,x)}\min\limits_{\substack{X\in\mathrm{T}_x\Sigma\\h_0[x](X,X)=1}}h_t[x](X,X)\right]>0,\quad\forall x\in\Sigma.
$$
If the continuous Riemannian manifold $(\Sigma,\Omega^2h_0)$ is complete then every hypersurface $\Sigma_t$ is a Cauchy surface, $t\in\mathbb{R}$.
\end{proposition}
\begin{proof} Note that $\Omega:\Sigma\to\mathbb{R}_+$ is a continuous function. Denoting $Y=\Omega(x)^{-1}X$, we observe that
$$
\frac{C(t)}{\beta^2(t,x)}\min\limits_{\substack{X\in\mathrm{T}_x\Sigma\\h_0[x](X,X)=1}}h_t[x](X,X)=\frac{C(t)}{\beta^2(t,x)}\min\limits_{\substack{Y\in\mathrm{T}_x\Sigma\\\Omega^2(x)h_0[x](Y,Y)=1}}\Omega^2(x)h_t[x](Y,Y)
$$
$$
=\frac{C(t)\Omega^2(x)}{D(t,x)}\ge\Omega^2(x),\quad\forall x\in\Sigma,
$$
implying that $D(t,x)\le C(t)$ for all $(t,x)\in\mathbb{R}\times\Sigma$ ($D$ being defined in terms of the complete metric $h_\infty=\Omega^2h_0$). It follows that $D$ is bounded on every cylindrical region $[t_1,t_1]\time\Sigma$, and hence on every region of the form $\mathrm{J}^+(t,x)\cap(-\infty,t_0]$ or $\mathrm{J}^-(t,x)\cap[t_0,+\infty)$. It remains to apply Proposition \ref{0xSigmaCauchyProp}. $\Box$
\end{proof}

\begin{remark} Again, the choice of $h_0$ as a reference metric is arbitrary, and any fixed $h_t$ for $t\in\mathbb{R}$ will serve equally well. Also, no restrictions are put on the long time behaviour of the function $C$ in contrast to Theorem 2.1 of \cite{ChoquetBruhatCotsakis2002}.
\end{remark}

\section*{Homogeneous cosmological spacetimes}

A homogeneous cosmological model assumes a certain amount of additional structure beyond the initial setting of General Relativity. One standard assumption is that a group of isometries acts on the spacetime, so that orbits of this action foliate the spacetime with smooth spacelike homogeneous Cauchy hypersurfaces. The actions of the isometry group on different hypersurfaces are assumed to be isomorphic, and the factorization $M\simeq\mathbb{R}\times\Sigma$ coming from global hyperbolicity is expected to be also equivariant. In our observations, the approach in mathematical physics literature to this wish-list of properties has mostly been "assume and go" (including the author's earlier works), and it is interesting to see if a minimal set of assumptions can be set forth with the rest being rigorously deduced.

Let $(M,\mathrm{g})$ be a connected smooth $d+1$-dimensional Lorentzian manifold with or without boundary, and let $\mathrm{Iso}(M)$ be the Lie group of its isometries. Let $G\subset\mathrm{Iso}(M)$ be a Lie subgroup. We set forth the following assumptions on $M$ and $G$:

\begin{itemize}

\item[1.] $G$ acts {\bf properly} on $M$.

\item[2.] The Killing vector fields given by the $G$-action are {\bf spacelike}, and the {\bf maximal rank} of the distribution they generate is $d$.

\item[3.] All $G$-orbits are connected, i.e., $Gx=G_0x$ for $\forall x\in M$, where $G_0\subset G$ is the identity component.

\end{itemize}

Proper action is an essential restriction that wards against many pathologies \cite{BerndtDiazVanaei2017}, \cite{RudolphSchmidt2013}. It amounts to the map $G\times M\ni(g,x)\mapsto (gx,x)\in M\times M$ being proper, i.e., the preimage of every compact set in $M\times M$ is compact in $G\times M$.

\begin{definition} We call an $M$ as above a ($G$-)homogeneous cosmological spacetime.
\end{definition}

\noindent Note that, aside from the proper action, these properties are easy to verify in practice. We denote by $\operatorname{q}:M\to M/G$ the quotient map $M\ni x\mapsto Gx\in M/G$. The proof of the following proposition relies on the fact that the action of $G$ on $M$ is cohomogeneity one (see \cite{BerndtDiazVanaei2017} for a discussion).

\begin{proposition}\label{GHomCosmSTProp} The following properties follow immediately from the definition of a $G$-homogeneous cosmological spacetime $M$.

\begin{itemize}

\item[1.] The quotient (orbit space) $I\doteq M/G$ is a connected smooth 1-dimensional manifold with or without boundary (thus diffeomorphic to either $\mathbb{R}$, $\mathbb{S}^1$, $[0,+\infty)$ or $[0,1]$), and the interior $\mathring I$ is diffeomorphic to either $\mathbb{R}$ or $\mathbb{S}^1$. Further, $\mathring M=\operatorname{q}^{-1}(\mathring I)$ and $\partial M=\operatorname{q}^{-1}(\partial I)$. Every interior point $Gx\in\mathring I$ is a generic orbit, while every boundary $Gx\in\partial I$ is exceptional or singular.

\item[2.] If $I$ is not diffeomorphic to $\mathbb{S}^1$ then there exists a fixed $d$-dimensional $G$-homogeneous space $\Sigma$, a smooth family $\mathbb{R}\ni t\mapsto h_t\in C^\infty(\mathrm{T}^2\Sigma)$ of $G$-invariant Riemannian metrics on $\Sigma$, and a smooth positive function $\beta\in C^\infty(\mathbb{R},\mathbb{R}_+)$, such that the Lorentzian $G$-manifold $(\mathring M,\mathrm{g})$ is isometrically and equivariantly isomorphic to the product $G$-manifold $(\mathbb{R}\times\Sigma,\beta^2\circ t\oplus-h_*)$.

\end{itemize}
\end{proposition}
\begin{proof} 1. By Proposition 6.1.5 and Proposition 6.3.4 in \cite{RudolphSchmidt2013}, the quotient space $I=M/G$ is locally compact, second countable, Hausdorff and connected (for so is $M$), and the quotient map $\operatorname{q}:M\to I$ is open. Moreover, $I$ is the disjoint union of finitely many smooth manifolds $I_\alpha$ corresponding to different orbit types (Proposition 6.6.1 in \cite{RudolphSchmidt2013}). Since $I$ is connected, there is a generic orbit type (principal stratum) $\alpha_*$, and the corresponding manifold $I_{\alpha_*}\subset I$ is open, dense and connected (Remark 6.6.2 in \cite{RudolphSchmidt2013}, Theorem 4.3.2 in \cite{Pflaum2001}). By Proposition 6.6.1 in \cite{RudolphSchmidt2013}, the restricted quotient map $\operatorname{q}_{\alpha_*}:M_{\alpha_*}\to I_{\alpha_*}$ is a submersion of the open dense (since $\operatorname{q}$ is open) $G$-invariant submanifold $M_{\alpha_*}=\operatorname{q}^{-1}(I_{\alpha_*})\subset M$.

By Corollary 6.3.5 in \cite{RudolphSchmidt2013}, each orbit $Gx\subset M$ is a closed embedded submanifold with a transitive $G$-action. The Killing vector fields of the infinitesimal $G$-action are parallel to $Gx$ and exhaust the tangent bundle $\operatorname{T}Gx$. The highest rank of the distribution generated by Killing vector fields corresponds to the highest dimension of a $G$-orbit, so that generic orbits, which by Theorem 4.3.2 in \cite{Pflaum2001} correspond to minimal stabilizer subgroups and thus have maximal dimension, are of dimension $d$. Thus $\dim I_{\alpha_*}=1$ and, being connected, $I_{\alpha_*}\simeq\mathbb{R}$ or $I_{\alpha_*}\simeq\mathbb{S}^1$. Since $I_{\alpha_*}\subset I$ is dense, manifolds $I_\alpha$ for all other orbit types $\alpha$ are connected zero dimensional manifolds, i.e., individual points. If $I_{\alpha_*}\simeq\mathbb{S}^1$ then $I=\mathring I=I_{\alpha_*}$ and $\partial I=\partial M=\emptyset$. If $I_{\alpha_*}\simeq\mathbb{R}$ then $I\simeq\mathbb{R}$, $I\simeq[0,+\infty)$ or $I\simeq[0,1]$, so that $\mathring I=I_{\alpha_*}$. The case $I\simeq\mathbb{S}^1$ with $I_{\alpha_*}\simeq\mathbb{R}$ cannot occur. Indeed, suppose that $I=\mathbb{S}^1$ and $I_{\alpha_*}=\mathbb{S}^1\setminus\{1\}$. Then $I_\epsilon\doteq\{e^{\imath\phi}\,|\,\phi\in(-\epsilon,\epsilon)\}\subset I$ is a connected open submanifold, and the $G$-action on the $G$-invariant open submanifold $M_\epsilon\doteq\operatorname{q}^{-1}(I_\epsilon)\subset M$ is still proper. By Theorem 4.3.2 in \cite{Pflaum2001} the principal stratum $I_{\epsilon,\alpha_*}\subset I_\epsilon$ is open, connected and dense, which implies $I_{\epsilon,\alpha_*}=I_\epsilon$, in contradiction with the fact that the orbit $q^{-1}(\{1\})$ is not generic. Finally, because the map $\operatorname{q}$ is continuous and open, we have $\operatorname{q}(\mathring M)\subset\mathring I$ and $\operatorname{q}^{-1}(\mathring I)\subset\mathring M$, whence $\mathring M=\operatorname{q}^{-1}(\mathring I)$.

2. We denote by $\Sigma$ a representative generic orbit with its $G$-action. Then $\mathring M$ is a principal $G$-bundle over $\mathring I$ (thus trivial) with standard fibre $\Sigma$ (Remark 6.6.2 in \cite{RudolphSchmidt2013}), and thus there is an isomorphism of principal $G$-bundles $\phi:\mathring M\to\mathbb{R}\times\Sigma$. The pushforward $\widetilde{\mathrm{g}}\doteq [d\phi^{-1}]^*\mathrm{g}$ of the Lorentzian metric $\mathrm{g}$ through $\phi$ is a $G$-invariant Lorentzian metric on $\mathbb{R}\times\Sigma$. Since Killing vector fields are spacelike, the metric on every orbit manifold $Gx$ induced from the spacetime metric $\mathrm{g}$ is Riemannian and $G$-invariant. Every $\Sigma_t=\{t\}\times\Sigma$ is the image $\phi(Gx)$ of an orbit $Gx$, thus the metric $\tilde h_t$ induced from $\widetilde{\mathrm{g}}$ on $\Sigma_t$ is Riemannian and $G$-invariant. The function $\widetilde{\mathrm{g}}(dt,dt)$ is thus $G$-invariant and therefore depends on $t$ only. Define $\beta:\mathbb{R}\to\mathbb{R}_+$ by
$$
\beta(t)\doteq\frac1{\sqrt{\widetilde{\mathrm{g}}^{-1}(dt,dt)}},\quad\forall t\in\mathbb{R}.
$$
The vector field
$$
X_0=\beta^2(t)\widetilde{\mathrm{g}}^{-1}(dt,.)-\partial_t
$$
is $G$-invariant and satisfies $dt(X_0)=0$, hence $X_0\in C^\infty(\mathbb{R}\times\mathrm{T}\Sigma)$. Every maximal integral curve $\gamma$ of $X_0$ is thus entirely in $\Sigma_t$ for some $t\in\mathbb{R}$. Then $X_0(t,.)\in C^\infty(\mathrm{T}\Sigma_t)$ is an invariant vector field on the homogeneous space $\Sigma_t$ and is thus complete. This shows that $\gamma$ is complete, and since $\gamma$ was arbitrary, $X_0$ is complete.

We denote by $s\mapsto e^{sX_0}$ the smooth 1-parameter group of diffeomorphisms associated with $X_0$ and define the $G$-equivariant diffeomorphism $\psi:\mathbb{R}\times\Sigma\to\mathbb{R}\times\Sigma$ by setting $\psi(t,p)=e^{-tX_0}(t,p)=(t,e^{-tX_0(t,.)}p)$ for all $(t,p)\in\mathbb{R}\times\Sigma$. Because $\psi$ is $G$-equivariant and acts along $\Sigma$, $h_t\doteq[d\psi(t,.)^{-1}]^*\tilde h_t$ is a Riemannian $G$-invariant metric on $\Sigma_t$ induced from $[d\psi^{-1}]^*\widetilde{\mathrm{g}}$. Then
$$
d\psi^{-1}\partial_t=\partial_t+X_0=\beta^2(t)\widetilde{\mathrm{g}}^{-1}(dt,.)
$$
and
$$
[d\psi^{-1}]^*\widetilde{\mathrm{g}}(\partial_t,\partial_t)=\widetilde{\mathrm{g}}(d\psi^{-1}\partial_t,d\psi^{-1}\partial_t)=\beta^4(t)\widetilde{\mathrm{g}}(\widetilde{\mathrm{g}}^{-1}(dt,.),\widetilde{\mathrm{g}}^{-1}(dt,.))=\beta^2(t),
$$
while
$$
[d\psi^{-1}]^*\widetilde{\mathrm{g}}(\partial_t,X)=\widetilde{\mathrm{g}}(d\psi^{-1}\partial_t,d\psi^{-1}X)=
$$
$$
\beta^2(t)\widetilde{\mathrm{g}}(\widetilde{\mathrm{g}}^{-1}(dt,.),d\psi^{-1}X)=dt(d\psi^{-1}X)=0,\quad\forall X\in C^\infty(\mathbb{R}\times\mathrm{T}\Sigma),
$$
because $\psi(\Sigma_t)=\Sigma_t$ and hence $d\psi^{-1}\mathrm{T}\Sigma_t=\mathrm{T}\Sigma_t$. This shows that $[d\psi^{-1}]^*\widetilde{\mathrm{g}}=\beta^2\circ t\oplus-h_*$, as desired. $\Box$
\end{proof}

Proper action of $G$ on $M$ implies in particular that all stabilizer subgroups $H\subset G$ are compact. Conversely, the product of every connected homogeneous Riemannian manifold $(G,\Sigma)$ that has compact stabilizers with $\mathbb{R}$ gives a $G$-homogeneous cosmological spacetime.

\begin{proposition} Let $\Sigma=G/H$ be a connected $G$-homogeneous space with $H\subset G$ a compact subgroup, and let $\mathbb{R}\ni t\mapsto h_t\in C^\infty(\mathrm{T}^2\Sigma)$ be a smooth 1-parameter family of $G$-invariant Riemannian metrics on $\Sigma$ and $\beta\in C^\infty(\mathbb{R},\mathbb{R}_+)$ a smooth positive function. Then the Lorentzian $G$-manifold $(\mathbb{R}\times\Sigma,\beta^2\circ t\oplus-h_*)$ is a $G$-homogeneous cosmological spacetime.
\end{proposition}
\begin{proof}If $\Sigma$ is connected then so is $\mathbb{R}\times\Sigma$, and (since the $G$-action is effective) $G$ can be considered as a Lie subgroup of $\mathrm{Iso}(\mathbb{R}\times\Sigma,\beta^2\circ t\oplus-h_*)$. First let us convince ourselves that the $G$-action on $\Sigma$ is proper. Because $H$ is compact, the canonical quotient map $\operatorname{q}:G\to\Sigma$ has the property that the preimage $\operatorname{q}^{-1}(K)\subset G$ of every compact set $K\Subset\Sigma$ is compact. Indeed, by Lemma 2.48 in \cite{Folland2015} there exists a compact $E\Subset\operatorname{q}^{-1}(K)$ such that $\operatorname{q}(E)=K$. But then $\operatorname{q}^{-1}(K)=E\cdot H$, and the product of two compact subsets is compact. We will now use Corollary 6.3.3 in \cite{RudolphSchmidt2013}. Let $\{a_n\}_{n=1}^\infty\subset G$ and $\{g_nH\}_{n=1}^\infty\subset G/H=\Sigma$ such that $g_nH\to g_0H\in\Sigma$ and $a_ng_nH\to g_\star H\in\Sigma$. Then exist compacts $K_1,K_2\Subset\Sigma$ such that $g_nH\in K_1$ and $a_ng_nH\in K_2$ for $n>>1$. This means that $g_n\in\operatorname{q}^{-1}(K_1)$ and $a_ng_n\in\operatorname{q}^{-1}(K_2)$, and hence $a_n\in\operatorname{q}^{-1}(K_2)\cdot\operatorname{q}^{-1}(K_1)^{-1}$ for $n>>1$. But $\operatorname{q}^{-1}(K_2)\cdot\operatorname{q}^{-1}(K_1)^{-1}$ is compact, which shows that $\{a_n\}_{n=1}^\infty$ has a convergent subset. This implies that the $G$-action on $\Sigma$ is proper. Now by Remark 6.3.9 in \cite{RudolphSchmidt2013}, the action of $G$ on the product manifold $\mathbb{R}\times\Sigma$ is also proper. The Killing vector fields generated by $G$ on $\mathbb{R}\times\Sigma$ are parallel to $\Sigma$ and thus spacelike, and the rank of the distribution they generate is the dimension of $\Sigma$ which is $\dim(\mathbb{R}\times\Sigma)-1$. $\Box$
\end{proof}

Let us finally turn to the question of global hyperbolicity of a homogeneous cosmological spacetime. The boundary of the spacetime does not play a role in the question of global hyperbolicity, so let us assume that $\partial M=\emptyset$. Then $M=\mathring M$ and $I=\mathring I$, the latter being diffeomorphic to either $\mathbb{R}$ or $\mathbb{S}^1$. In the latter case $M$ is clearly not globally hyperbolic, because it possesses a closed timelike curve. Therefore, without loss of generality, let us assume that $(M,\mathrm{g})=(\mathbb{R}\times\Sigma,\beta^2\circ t\oplus-h_*)$ is our $G$-homogeneous cosmological spacetime. We want to show that the foliation by orbits is also a foliation by Cauchy hypersurfaces, so that the wish-list of desired properties is fulfilled.

\begin{proposition} A $G$-homogeneous cosmological spacetime $(\mathbb{R}\times\Sigma,\beta^2\circ t\oplus-h_*)$ as above is globally hyperbolic, and every hypersurface $\Sigma_t$ is a Cauchy surface.
\end{proposition}
\begin{proof} As a Riemannian homogeneous space, $(\Sigma,h_0)$ is complete. The continuous function $D:\mathbb{R}\times\Sigma\to\mathbb{R}_+$ from (\ref{DDef}) with $h_\infty=h_0$ and $\beta\circ t$ is clearly $G$-invariant, and thus $D(t,x)=D(t,x_0)$ for a fixed $x_0\in\Sigma$ and all $(t,x)\in\mathbb{R}\times\Sigma$. Now Proposition \ref{COmegaProp} is applicable with $C(t)=D(t,x_0)^{-1}$ and $\Omega(x)=1$. $\Box$
\end{proof}

\section*{Homogeneous cosmological vector bundles}

Having set up the geometry of the spacetime $(M,\mathrm{g})$, the next step is to configure matter fields propagating on it. A state of a classical field is normally given by a section in a vector bundle $\mathcal{T}\to M$, and the subject of the present section is the geometry and topology of vector bundles appearing in homogeneous cosmological models. We saw that global hyperbolicity immediately implies a factorization $M\simeq\mathbb{R}\times\Sigma$, and if the spacetime is homogeneous cosmological then this factorization is also equivariant, meaning that $M$ can be substituted by its factorized form for essentially all purposes without loss of generality. The next natural question that should arise is when and whether the vector bundle $\mathcal{T}\to M$ over such a factorized spacetime also factorizes as a product $\mathcal{T}\simeq\mathbb{R}\times\mathcal{S}$ where $\mathcal{S}\to\Sigma$ is a standard vector bundle, and if a symmetry group $G$ acts on $\mathcal{T}$, whether or not such a factorization is equivariant. We believe that these questions have not been addressed in the mathematical physics literature, and a good reason for that may have been the following circumstance. Once the spacetime has been factorized as a foliation $M\simeq\mathbb{R}\times\Sigma$ by constant time hypersurfaces isomorphic to a standard space $\Sigma$, all covariant vector bundles (e.g., tensor, spinor) are automatically factorized (equivariantly, if homogeneous cosmological) as $\mathcal{T}\simeq\mathbb{R}\times\mathcal{T}|_{\Sigma_0}$. However, modern literature on vector valued fields propagating on curved spacetimes begins with an arbitrary vector bundle over the spacetime \cite{BaerGinouxPfaeffle2007}, \cite{BrunettiDappiaggiFredenhagenYngvason2015}, and a wealth of very convenient technical tools based on global time arguments depend on our ability to extend the factorization to the general case. We will see that the factorization of the vector bundle holds in full generality, so that for essentially all purposes the vector bundle can be assumed to be factorized without loss of generality.

That the vector bundle $\mathcal{T}\to\mathbb{R}\times\Sigma$ is isomorphic to the product $\mathbb{R}\times\mathcal{T}|_{\Sigma_0}$ is not difficult to show using the homotopy classification of vector bundles, e.g., Chapter 7 in \cite{Karoubi1978}. The less trivial question is whether the factorization is equivariant, given a group of spacelike symmetries. Various homotopy classifications of equivariant vector bundles have been known in topology literature for decades under various restrictions (compactness of $\Sigma$ or $G$, see \cite{Bierstone1973},\cite{Lashof1982} etc.). Thankfully, the exhausting answer in amazing generality became available recently in \cite{LueckUribe2014}, and the following proposition is a direct application of methods developed therein.

Denote for every fibre bundle $E\to\mathbb{R}\times\Sigma$ the slice over $\Sigma_0$ by
$$
E_0\doteq E|_{\Sigma_0}.
$$
For a Lie group $G$, a manifold $M$ with a smooth $G$-action, and a vector bundle $\mathcal{T}\overset{\pi}{\longrightarrow}M$, a smooth $G$-action on $\mathcal{T}$ covering the action on $M$ is a smooth $G$-action $\Phi:G\times\mathcal{T}\to\mathcal{T}$ on the total space $\mathcal{T}$ such that
$$
\pi(gy)=g\pi(y),\quad\Phi|_{\mathcal{T}_x}\in\mathrm{Hom}(\mathcal{T}_x,\mathcal{T}_{gx}),\quad\forall y\in\mathcal{T},\quad\forall x\in M,\quad\forall g\in G.
$$
\begin{proposition} Let $\Sigma$ be a connected manifold and $G$ a Lie group acting on $\Sigma$ with compact stability subgroups. Let further $\mathcal{T}\to\mathbb{R}\times\Sigma$ be a real or complex vector bundle with a smooth $G$-action covering the natural action on $\mathbb{R}\times\Sigma$. Then there exists an isomorphism of equivariant vector bundles
$$
\varphi:\mathbb{R}\times\mathcal{T}_0\to\mathcal{T}
$$
such that the restriction of $\varphi$ to $\{0\}\times\mathcal{T}_0$ is the projection onto second component.
\end{proposition}
\begin{proof} Consider the frame bundle $\mathcal{P}\to\mathbb{R}\times\Sigma$ of $\mathcal{T}$, which is a $G$-equivariant principal $\mathrm{GL}$-bundle ($\mathrm{GL}\in\{\mathrm{GL}(n,\mathbb{R}),\mathrm{GL}(n,\mathbb{C})\}$, $n=\dim\mathcal{T}$), in terms of Definition 2.1 in \cite{LueckUribe2014}. It is sufficient to establish an isomorphism of $G$-equivariant principal $\mathrm{GL}$-bundles
$$
\phi:\mathbb{R}\times\mathcal{P}_0\to\mathcal{P},\quad\phi|_{\{0\}\times\mathcal{P}_0}=\mathbf{1}.
$$
If $\mathcal{R}(\mathcal{P})$ is the family of local representations of $\mathcal{P}$ then the corresponding families of all principal bundles below will be subfamilies of $\mathcal{R}(\mathcal{P})$. Since $G$ is a Lie group and $\mathrm{GL}$ a connected Lie group, and because stability subgroups of the $G$-action on $\Sigma$ (and hence on $\mathbb{R}\times\Sigma$) are compact, by Theorem 6.3 in \cite{LueckUribe2014} the Condition (H) of Definition 6.1 in \cite{LueckUribe2014} is satisfied for $\mathcal{R}(\mathcal{P})$. Let now $p:E\to B$ be the universal $G$-equivariant principal $\mathrm{GL}$-bundle for $\mathcal{R}(\mathcal{P})$ from the homotopy classification Theorem 11.4 in \cite{LueckUribe2014}. Then there exists a $G$-map $f:\mathbb{R}\times\Sigma\to B$ such that
$$
\mathcal{T}=f^*E=\bigcup_{(t,x)\in\mathbb{R}\times\Sigma}\{(t,x)\}\times p^{-1}(f(t,x)).
$$
If we introduce $g:\mathbb{R}\times\Sigma\to B$ by $g(t,x)=f(0,x)$ for all $x\in\Sigma$ then
$$
\mathbb{R}\times\mathcal{T}_0=g^*E=\bigcup_{(t,x)\in\mathbb{R}\times\Sigma}\{(t,x)\}\times p^{-1}(f(0,x)).
$$
An explicit $G$-homotopy $F:\mathbb{R}\times\Sigma\times[0,1]\to B$ between $f$ and $g$ can be constructed as follows,
$$
F(t,x;\alpha)\doteq f(\alpha t,x),\quad\forall (t,x)\in\mathbb{R}\times\Sigma,\quad\forall\alpha\in[0,1],
$$
$$
F(t,x;0)=g(t,x),\quad F(t,x;1)=f(t,x).
$$
Consider the $G$-equivariant principal $\mathrm{GL}$-bundle
$$
F^*E=\bigcup_{(t,x;\alpha)\in\mathbb{R}\times\Sigma\times[0,1]}\{(t,x;\alpha)\}\times p^{-1}(f(\alpha t,x)).
$$
Observe that
$$
F^*E|_{\mathbb{R}\times\Sigma\times\{0\}}=\mathbb{R}\times\mathcal{T}_0,\quad F^*E|_{\mathbb{R}\times\Sigma\times\{1\}}=\mathcal{T}.
$$
By Theorem 10.1 of \cite{LueckUribe2014} there exist isomoprhisms
$$
\phi_0:\mathbb{R}\times\mathcal{T}_0\times[0,1]\to F^*E,\quad\phi_0|_{\mathbb{R}\times\mathcal{T}_0\times\{0\}}=\mathbf{1},
$$
$$
\phi_1:\mathcal{T}\times[0,1]\to F^*E,\quad\phi_1|_{\mathbb{R}\times\mathcal{T}_0\times\{1\}}=\mathbf{1},
$$
and they yield an isomorphism
$$
\phi_1^{-1}\circ\phi_0:\mathbb{R}\times\mathcal{T}_0\times[0,1]\to\mathcal{T}\times[0,1],
$$
which by restriction gives the isomorphism
$$
\phi_*\doteq[\phi_1^{-1}\circ\phi_0]|_{\mathbb{R}\times\mathcal{T}_0\times\{0\}}:\mathbb{R}\times\mathcal{T}_0\to\mathcal{T}.
$$
A further restriction provides an isomorphism
$$
\phi_\dagger\doteq\phi_*|_{\{0\}\times\mathcal{T}_0}:\mathcal{T}_0\to\mathcal{T}_0,
$$
and finally our desired isomorphism is
$$
\phi\doteq\phi_*\circ(\mathbf{1}\times\phi_\dagger^{-1}):\mathbb{R}\times\mathcal{T}_0\to\mathcal{T}.
$$
The proof is complete. $\Box$
\end{proof}

\begin{remark} The choice of the hypersurface $\Sigma_0$ is arbitrary, and the statement remains true for every $\Sigma_t$, $t\in\mathbb{R}$.
\end{remark}

\subsection*{Acknowledgements}

The last part of this work devoted to vector bundles would be impossible without the very enlightening discussions with Nikolai Saveliev and Fedor Manin, whom we express our deep gratitude. We further thank the anonymous referee for their suggestions and for pointing out an inconsistency in part 2. of Proposition \ref{GHomCosmSTProp} in the original version of the manuscript, which lead to a substantial improvement of the text.

\end{document}